\documentclass[letterpaper, 10 pt, conference]{ieeeconf}
\usepackage{comment}
\usepackage{amssymb, amsmath,graphicx,charter, latexsym}
\usepackage{amsfonts}
\usepackage{bbm}
\usepackage{algorithm}
\newtheorem{definition}{Definition}
\newtheorem{lemma}{Lemma}

\newtheorem{theorem}{Theorem}

\usepackage{graphicx}
\usepackage{url}
\usepackage{epstopdf}

\IEEEoverridecommandlockouts                              
\def\bs{\ensuremath\boldsymbol}
\begin{document}
\title{Dynamic Adaptive Streaming using Index-Based Learning Algorithms 
}

\author{Rahul Singh and P. R. Kumar 
\thanks{R. Singh is at 32-D716, LIDS, MIT, Cambridge, MA 02139; P. R. Kumar is at 
Dept. of ECE, Texas A\&M Univ., 3259 TAMU, College Station, TX 77843-3259.
{\tt\small rsingh12@mit.edu, prk@tamu.edu.}
 } %
 }

\maketitle
\IEEEpeerreviewmaketitle

\begin{abstract}
We provide a unified framework using which we design scalable dynamic adaptive video streaming algorithms based on index based policies (dubbed \textbf{DAS-IP} Fig.~\ref{figsoln}) to maximize the Quality of Experience (QoE) provided to clients using video streaming services. Due to the distributed nature of our algorithm DAS-IP, it can be easily implemented in lieu of popular existing Dynamic Adaptive Streaming over HTTP (DASH) algorithm which is used by various Cloud based video streaming services, Content Delivery Networks (CDNs), Cache networks, wireless networks, vehicular networks etc.

We begin by considering the simplest set-up of a one-hop wireless network in which an Access Point (AP) transmits video packets to multiple clients over a shared unreliable channel. The video file meant for each client has been fragmented into several packets, and the server maintains multiple copies (each of different quality) of the same video file.
Clients maintain individual packet buffers in order to mitigate the effect of uncertainty on video iterruption. Streaming experience, or the Quality of Experience (QoE) of a client depends on several factors: i) starvation/outage probability, i.e., average time duration for which the client does not play video because the buffer is empty, ii) average video quality, iii) average number of starvation periods, iv) temporal variations in video quality etc.

We pose the problem of making dynamic streaming decisions in order to maximize the total QoE as a Constrained Markov Decision Process (CMDP). A consideration of the associated dual MDP suggests us that the problem is vastly simplified if the AP is allowed to charge a price per unit bandwidth usage from the clients. More concretely, a ``client-by-client" QoE optimization leads to the networkwide QoE maximization, and thus provides us a decentralized streaming algorithm.

This enables the clients to themselves decide the optimal streaming choices in each time-slot, and yields us a much desired client-level adaptation algorithm. The optimal policy has an appealing simple threshold structure, in which the decision to choose the video-quality and power-level of transmission depends solely on the buffer-level. In case the clients are unaware of their (possibly random and time-varying) system parmeters, we develop algorithms that learn the indices while utilizing the strucure of the optimal decentralized policy. The decentralized nature of optimal policy implies that the DAS-IP has a much ``smaller" policy space to explore from, and hence converges fast.
\begin{figure}[!t]
	\centering
	\includegraphics[width=0.5\textwidth]{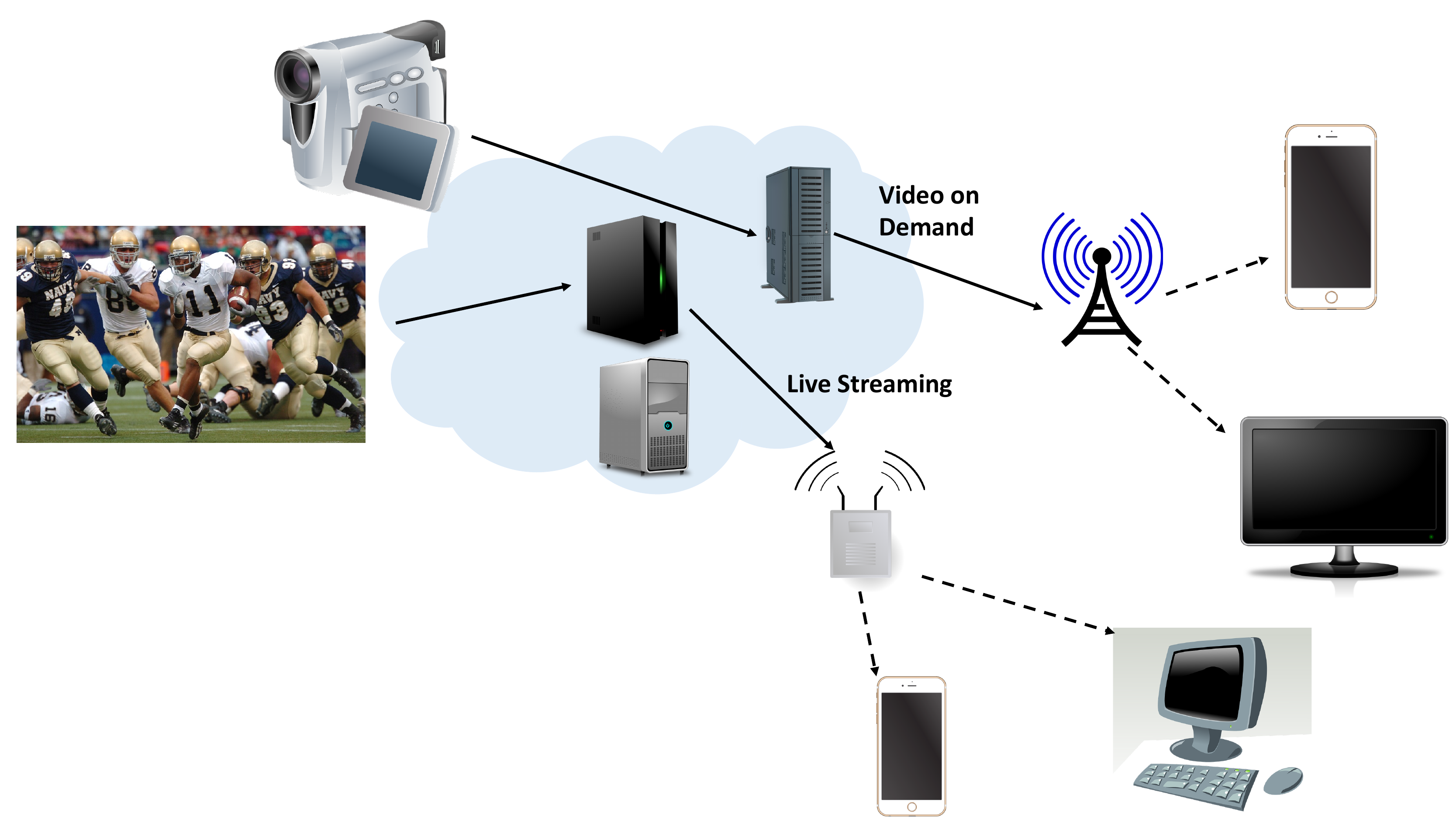}
	\caption{A Cloud live streaming a game, and hosting a Video on Demand (VoD) service. A continuous adaptation of bit rate of \emph{each} user is required in order to ensure a high level of Quality of Experience to end users.}
	\label{figcloud}
\end{figure}
\begin{figure}[!t]
	\centering
	\includegraphics[width=0.5\textwidth]{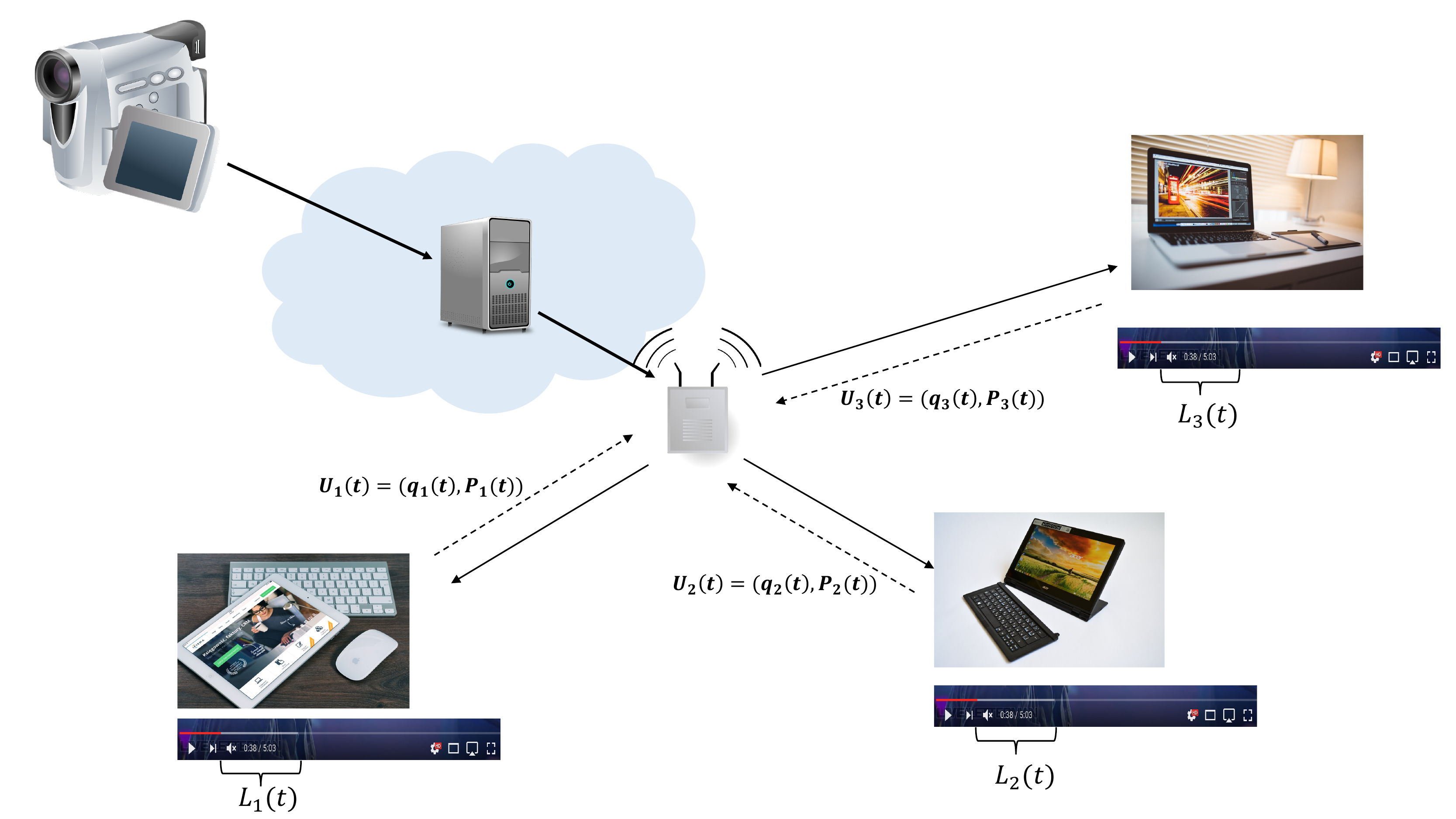}
	\caption{\textbf{DAS-Index Policy} (Dynamic Adaptive Streaming using Index Policy): Under the proposed solution, clients observe their instantaneous video quality, playback buffer level $L(t)$, wireless channel conditions and request the ``optimal" packet delivery rate $U_n(t)$, which is characterized by the transmission power level $P_n(t)$, and video quality $q_n(t)$.}
	\label{figsoln}
\end{figure}    
\end{abstract}
\section{Introduction}
Mobile video traffic accounted for 55 $\%$ of total mobile data traffic in 2015, and the share is expected to increase to 75 $\%$ by 2020. Unlike traditional QoS metrics such as throughput or delay, the user experience for video streaming applications depends on several complex factors, and is in itself an active area of research~\cite{Balachandran:2013:DPM:2534169.2486025,Dobrian:2011:UIV:2018436.2018478,Shafiq:2014:UIN:2637364.2591975}. In order to meet the stringent Quality of Experience (QoE) requirements imposed by video streaming applications, service providers have switched to advanced platforms such as Cloud based services~\cite{5754008}, Content Delivery Networks (CDNs)~\cite{Pallis:2006:IPC:1107458.1107462} etc. Moreover, they also use adaptive bitrate streaming algorithms such as DASH~\cite{dash}, HTTP Live Streaming (HLS) in order to continually monitor and improve the streaming experience.
      
Take the example of popular cloud services such as Microsoft Azure, IBM Cloud, Google Cloud, Amazon CloudFront, Apple's iCloud that provide live-streaming, on-demand video, online gaming services etc (Fig.~\ref{figcloud}. A party subscribing to live video streaming service with such a cloud will generate video file and upload it to the cloud in real-time. The cloud transcodes this data into multiple bit-rates, and the audience of this particular stream are served the video file using DASH. DASH enables a viewer to switch to low quality video in case his connection bandwidth is low, thus avoiding video interruptions. Since a major chunk of video data is demanded by mobile devices (that typically have bandwidth fluctuations), this enables the streaming service to reach a wider range of audiences. Since the CPU and storage-intensive video tasks get shifted to the cloud, this also enables computing-limited devices such as smartphones, access to high-quality, high-definition video in a wide range of locations. Figure~\ref{figcloud} depicts such a cloud service.

However, the state-of-the art adaptive streaming algorithms are unable to provide a satisfactory Quality of Experience (QoE) video streaming. As an example, the popular DASH algorithm is either too slow to respond to changes in congestion levels, or it is overly sensitive to short-term network bandwidth variations~\cite{a5}. Similarly, for clients served over wireless networks~\cite{a8}, rate adaptation needs to take complex factors such as channel fading into account. Experimental studies into studying the rate adaptation techniques employed by popular DASH clients such as Microsoft Smooth Streaming~\cite{microsoft}, Adobe OSMF~\cite{adobe}, Netflix have demonstrated that these algorithms perform poorly. 

Thus, the focus of this paper is to develop adaptive streaming algorithms that optimize the user experience by taking various factors into account while making streaming decisions. As will be shown later, our algorithms are decentralized, easily implementable and can be computed in a distributed fashion. This enables them to be embedded into existing techniques such as DASH. 

As a side-product of our analysis, we also touch upon a somewhat related problem of pricing the service resources so as to maximize the operator's revenue.  
\subsection{Past Works}
Previous works on video streaming have analyzed various relevant trade-offs. Trade-off between outage probability and number of initially buffered packets (initial delay time) is analyzed in~\cite{ParandehGheibi,gliang,egger,a4}, while~\cite{yim} studies the effect of variations in the temporal quality of videos on the global video quality.~\cite{Yuedong} studies the impact of flow level dynamics (flows entering and leaving the system) on the streaming Quality of Experience (QoE).~\cite{a5} considers the problem of controlling the rate at which a single client requests data from the server in order that the requested rate closely floows the TCP throughput available to it. However, the model assumes that only a single client is present in the network, ignores inherent system randomess and proposes a heuristic scheme.~\cite{6913491} provides an extensive survey on the QoE related works from human computer interaction and networking domains.

References \cite{rahul,rahul1,Rahul2015,guosingh,rs,rs1} develop a framework to design policies which provide services to clients in a regular fashion, though not in the context of video streaming QoE.
\subsection{Challenges}
Several metrics such as starvation probability (average time spent without video streaming), start-up delay, time spent in rebuffering, average video quality~\cite{Dobrian:2011:UIV:2018436.2018478}, temporal quality variations~\cite{yim}  etc. collectively decide the end user streaming experience. 
Thus, an optimized streaming necessarily involves achieving optimal trade-offs between these competing metrics. As an example, a higher time spent rebuffering (and hence increased delay before video begins) leads to a lesser playback interruptions~\cite{6567080}. A lower video quality increases transmission rate, and hence reduces the starvation probability. Though dynamically switching between different qualities definitely improves upon a combination of packet starvation-video quality, it also introduces temporal variations in video quality, which are known to affect streaming experience.

Other than these trade-offs, the network dynamics also need to be taken into account~\cite{Shafiq:2014:UIN:2637364.2591975}. For example, the algorithm should switch to a low quality video upon detecting a reduction in the bandwidth of a client. Similarly, it needs to re-allocate streaming resources upon an entry or exit of clients with time.
\subsection{Our Contribution}
Contrary to the vast existing literature which provides no theoretical guarantees on the QoE properties of the proposed schemes, our algorithms maximize the combined QoE of all the users in the network. As an example,~\cite{a5} devises policy to minimize interruptions for a single client. However, a networkwide deployment of such a policy at each client need not maximize the combined QoE, i.e., a client-by-client optimization need not maximize the overall QoE experienced by the combined set of all clients~\cite{ho1972team}.

Moreover, previous works address at most one particular aspect of the optimal streaming problem, in contrast we devise algorithms that optimize the overall experience by taking into account all the factors that decide final end user QoE. 

We begin by considering the simple scenario of single last-hop case for ease of exposition, and pose the problem of maximizing the combined streaming experience (QoE) of $N$ clients as a Constrained MDP (CMDP). We then consider the associated dual problem, and show that it yields a simple, decentralized solution. This allows the clients to themselves decide upon the optimal streaming decisions.

In case the clients do not know their system parameters, e.g. wireless channel's reliability or fading model, then we develop a simple online Reinforcement Learning~\cite{sutton} algorithm using which they can ``infer" the optimal streaming decisions. Since we are able to show that the optimal policy lies in the space of decentralized policies, we are able to cut down drastically (from exponential in $N$ to linear) on the policy space that needs to be explored. Consequently, the learning occurs at a client level, and converges fast. 
\section{System Description}\label{sec:sd}
\emph{Video File Description} The server has multiple copies (files) of the same video, where each file has a different video quality or bitrate. Each such video file has been broken into a sequence of small HTTP-based file segments. While the complete video file is typically several hours in duration (e.g. movie or a live broadcast), each segement's playback time is of the order of few seconds~\cite{a5}. Such an assumption is common in popular adaptive bitrate streaming techniques such as DASH~\cite{dash}.

\begin{figure}[!t]
	\centering
	\includegraphics[width=0.5\textwidth]{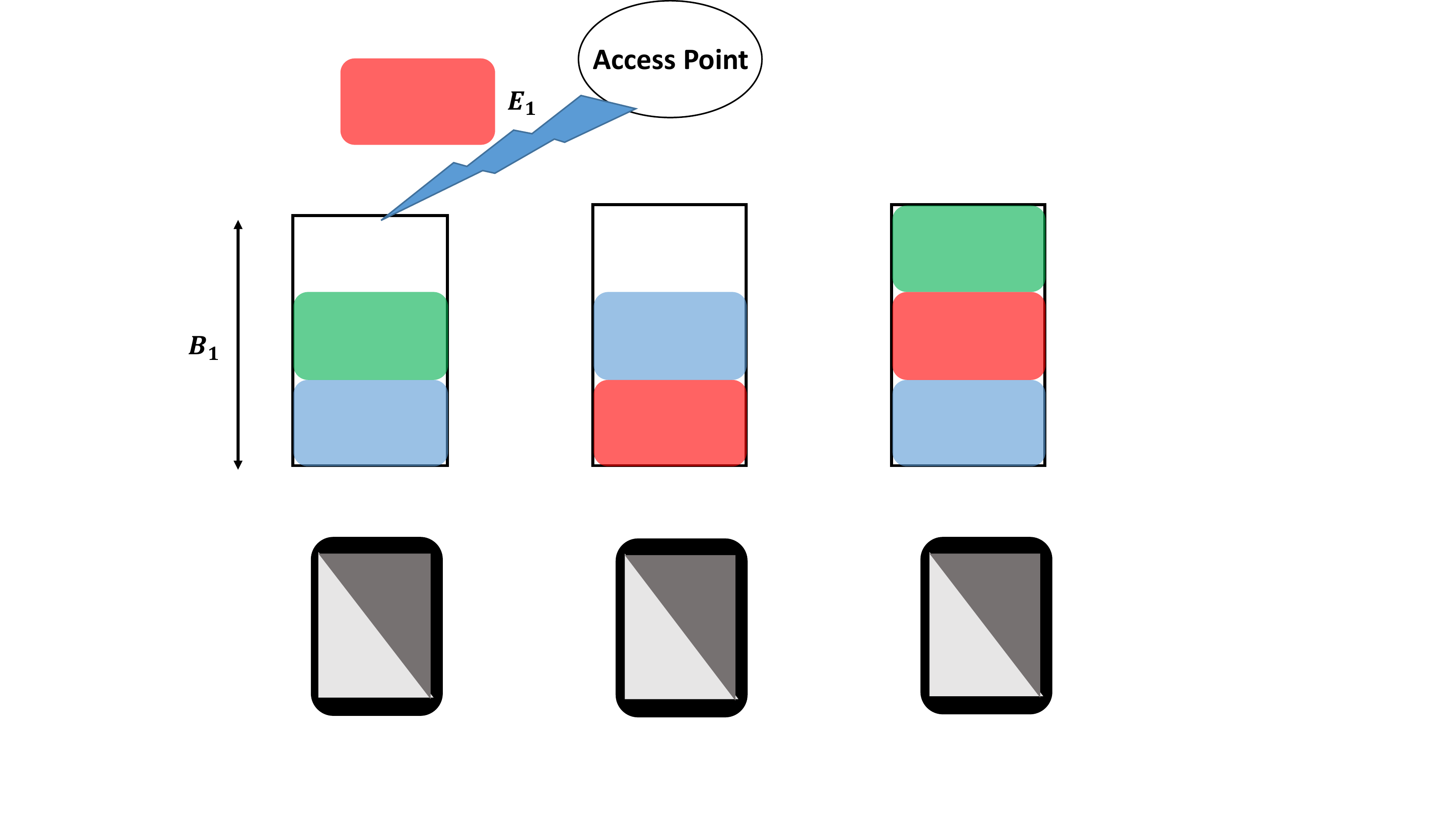}
	\caption{Clients video streaming packets from an Access Point over a shared wireless channel. $B$ denotes the buffer-size, while different colours denote packets of different video qualities.}
	\label{fig1}
\end{figure}
\emph{Network Description} 
In this section, we exclusively focus on the case of a single hop wirless network serving $N$ clients. We choose to present this simple case first In order to simplify the exposition and introduce the key concepts without resorting to complex notations. Later Sections will deal with more complex scenarios. Moreover, we ignore the interaction of our proposed algorithms with the TCP. We assume that the ``bottleneck link" is the last single hop link connecting the client to access point (AP).

A single wireless channel is shared by $N$ clients. Time is discretized, and the AP can attempt a packet transmission in a single time-slot. The buffer size of client $n$ is $B_n$ packets, and it plays a single packet for a duration of $T_n$ time-slots, see Fig.~\ref{fig1}. After finishing streaming of a packet, it starts streaming the packet at the head of its buffer. However, if it finds that the buffer is empty, then the streaming is interrupted, thus causing an ``outage". This event is also called ``starvation", i.e., a client is ``starved" of packets for sufficiently long duration, thereby causing video interruptions.
 
The clients are connected through unreliable wireless channels, so that the packet transmissions are assumed to be random. For simplicity, we will exclusively deal with the case of i.i.d. channel states\footnote{Fading channels will be covered in a later Section~\ref{fading}.}. The scheduler can pick the packet quality from $\{1,2,\ldots,Q_n\}$. A lower index of quality class is associated with a higher streaming experience. Thus, the ``cost" or disutility incurred by client $n$ for streaming a packet of quality $q$ is equal to $\lambda_{q,n}$.

The scheduler has the option to choose the power/energy level used for transmitting client $n$'s packets from the set $\{E_1,E_2,\ldots,E_M\}$. Throughout $E_1=0$ and referes to the case  of no transmission. The link reliabilities depend upon the transmission power, and the packet file size (in bits), so that a client $n$ packet of quality $q$, which is transmitted at power $E$, is delivered with a probability $P_n(q,E)$.

\emph{Streaming Experience Cost}
The streaming quality eperienced by a single client is assumed to depend upon
\begin{enumerate}
\item The average number of outages.
\item How ``often" the video gets interrupted, i.e., the number of outage-periods which is equal to the number of time-slots characterized by a  ``non-outage" to outage transition. 
\item Temporal variations, i.e., the number of times the video-quality is changed.
\item Average video quality which in turn is measured by the average disutility associated with the different video quality types.
\end{enumerate}
\section{Problem Formulation}
We begin with some notation. For simplicity, we will restrict ourselves to optimizing the factors 1)-3) listed in the streaming cost. Objective 4) will be dealt with separately in a later section.

Let $O_n(s)$ be the random variable that is $1$ if the $n$-th client faces an outage at time $s$, and $0$ otherwise, and thus is the indicator random variable of the outage event.
Let $\hat{E}_n(s)$ be the transmission power of client $n$ at time-slot $s$ and let $I_{q,n}(s)$ be the indicator random variable of the event that a packet of quality $q$ is delivered to client $n$ in time-slot $s$.
For simplicity, we will begin with a consideration of the scheduling problem under an average power constraint on AP, which models the case when the AP is operated by a battery that needs periodical charging. Streaming under constraints on the peak power, and on the number of orthogonal/independent channels which can be utilized for concurrent packet transmissions will be subject of later sections.   

The streaming problem is to the following Constrained Markov Decision Process (CMDP)~\cite{altman} for the optimal policy $\pi^\star$,
\begin{align}
&\min_{\pi} \lim_{t \to \infty}\frac{1}{t}\mathbb{E}\sum_n\sum_{s}\left( O_n(s)+\sum_{q=1}^{Q_n}\lambda_{q,n} I_{q,n}(s) \right. \label{costavg}\\
&\qquad \left.  \vphantom{\left( O_n(s)+\sum_{q=1}^{Q_n}\lambda_{q,n} I_{q,n}(s) \right)}+ \lambda_{O,n} \lvert O_n(s)\left(O_n(s-1)-1\right)\rvert \right)\notag\\
& \qquad\mbox{ subject to },\\
&\lim_{t\to\infty}\frac{1}{t}\mathbb{E}\left\{\sum_n\sum_{s}\hat{E}_n(s) \right\}\leq \bar{E}\label{pmdp}.
\tag{Primal MDP}
\end{align}
The term $|O_n(s)\left(O_n(s-1)-1\right)| $ counts the number of outage periods faced by client $n$ since it takes the value $1$ only if time-slot $s$ is the beginning of an outage-period for client $n$. The weighing parameters $\left\{ \lambda_{q,n}\right\}_{q=1}^{Q_n},\lambda_{O,n}\quad n=1,2,\ldots,N$ can be used to tune the QoS since they decide the relative importance that is placed on the different objectives. 

Since the video quality of a packet belonging to a higher class is less, we have
\begin{align*}
i>j \implies \lambda_{i,n}> \lambda_{j,n} \forall n=1,2,\ldots,N.
\end{align*}
The control action $u$ is a vector that describes the transmission power and video packet quality utilized for each user $n$. Let us denote by $C_n(l,u)$ the expected value of single step cost,
$$
O_n(s)+\sum_{q=1}^{Q_n}\lambda_{q,n} I_{q,n}(s) + \lambda_{O,n} \lvert O_n(s)\left(O_n(s-1)-1\right)\rvert, 
$$
when the buffer level of the client $n$ at time $s-1$ is $l$, and a control action $u$ is applied to it at time $s-1$. 
Thus the above problem is a CMDP in which the system state at time $t$ is described by the $N$ dimensional vector $L(t):=\left(l_1(t),l_2(t),\ldots,l_N(t)\right)$, where $l_n(t)$ is the amount of play time remaining in the buffer of client $n$ at time $t$. It can be posed as a linear program~\cite{altman} in which the variable is the steady-state measure $\mu_\pi$ induced by a policy $\pi$ on the joint state-action space $(X,U)$. Under this approach, the average cost~\eqref{costavg} is treated as the dot product between the measure $\mu_{\pi}(\cdot,\cdot)$ and the one step cost function $C(\cdot)$. The number of variables in this LP is equal to the cardinality of the joint state-action space, and hence increases exponentially with the number of clients $N$.

\emph{Linear Program for MDP}

We will however take a different approach to solving~\eqref{pmdp}. We will utilize the fact that the problem is convex (because it can be posed as an LP) and consider solving the associated dual problem. The consideration of the dual problem will show us that the original problem of scheduling $N$ clients decomposes into $N$ separate problems, where each problem now involves only a single client. In effect, the problem complexity grows linearly in the number of clients. Additionally, the optimal policy can be implemented in a decentralized fashion and computed using distributed updates~\cite{tsitsiklis1986distributed}. Moreover, an analysis of the single client MDP shows that the optimal policy has a simple threshold structure, and hence is easy to implement.   
\section{The Dual MDP}
We begin with some notation. We will use $\bar{C}_n$ to denote the time-average cost incurred by client $n$ under a stationary policy $\pi$, thereby suppressing the dependence on policy. Similarly $\bar{E}_n$ will be the average power consumption of client $n$. 

Letting $\lambda_E$ be the Lagrangian multiplier associated with the average power constraint $\sum_n \bar{E}_n\leq E$, the Lagrangian for the problem~\eqref{pmdp} is given by,
\begin{align}\label{lagrange}
\mathcal{L}(\pi,\lambda_E) = \sum_n \bar{C}_n+ \lambda_{E}\left(\bar{E}_n(s)-\bar{E}\right),
\end{align}
while the dual function is given by,
\begin{align}\label{dual}
D(\lambda_E)= \min_{\pi}\mathcal{L}(\pi,\lambda_E),
\end{align}
and the dual problem is stated as,
\begin{align}\label{dualprob}
\max_{\lambda_E\geq 0} D(\lambda_E).
\end{align}
We realize that the Lagrangian~\eqref{lagrange} decomposes into the sum of individual costs $\bar{C}_n+\lambda_E \bar{E}_n$ incurred by each client $n$, and hence in order to compute the dual function $D(\lambda_E)$, each of these individual client costs could be optimized separately by designing $\pi_n$, the policy for client $n$. This decomposition is key to the attractive properties of our proposed policies that we mentioned in the previous section. Thus,
\begin{lemma}\label{l4}
$D(\lambda_E) = \sum_n V_n(\lambda_E)-\lambda_E \bar{E}$.
\end{lemma}
Since the policies $\pi_n$ can be combined in straightforward manner (by implementing $\pi_n$ for each client $n$) in order to obtain a policy $\pi=\otimes \pi_n$ for the overall system, we shift our focus to solving the optimal policy for a single client that minimizes its cost $\bar{C}_n+\lambda_E \bar{E}_n$.
\section{Single Client Problem} \label{scp}
In this entire section, we will omit the sub-script $n$ with the understanding that all the quantities being referred to belong to client $n$.

Client has a buffer of capacity $B$ time-slots of play-time video in which stores video packets. Note that this assumption is equivalent to assuming a buffer of size $B$ packets because a packet gets played for $T$ time-slots.
In each time-slot $t$, the AP has to choose the following two control quantities for the client:
\begin{enumerate}[i)]
\item The video quality $ q(t) \in \{1,2,\ldots,Q\}$ of the packet transmitted to the client, and, 
\item the power $E(t)\in \{\hat{E}_1,\hat{E}_2,\ldots,\hat{E}_n\}$ utilized for the transmission.
\end{enumerate}
Let $l(t)$ be the play-time duration of the packets present in the buffer at time $t$, and hence it denotes the state of the client at time $t$. The wireless channel connecting the client to AP is random, and the distribution of the outcome is described as follows. 

A packet transmission of quality $q$ utilizing a power of $E$ units at time-slot $t$ is successful with a probability $P(q,E)$.
Thus, if $l(t)\leq  B-T+1$, then $l(t+1)$ is equal to $(l(t)-1)^{+}+T$ with a probability $P(q,E)$, while it assumes the value $(l(t)-1)^{+}$ with a probability $1-P(q,E)$. But if $l(t)>B-T+1$, then because acepting a new packet in the buffer will cause it to overflow, the system state at time $t+1$ is equal to $l(t)-1$ with a probability $1$.

Let  
\begin{align}\label{sf}
&\mathcal{S}(x) := 
\begin{cases}
(x-1)^{+}+T,\mbox{ if } x \leq B-T+1,\\
x-1,\mbox{ if } B-T+1<x\leq B, 
\end{cases}\\
&\mathcal{F}(x) := (x-1)^{+},
\end{align}
be the state values resulting from a successful and failed packet transmission respectively when the system state is $x$. Let $u(t):=(q(t),E(t))$ be the control action chosen by AP at time $t$, where $q(t),E(t)$ are the video quality and transmission power level chosen at time $t$.

\emph{Single Client Unit Step Cost }: The client is charge a cost of $\lambda_{E} \times E$ for transmitting a packet at power level $E$. It faces a penalty of $1$ units if there is an outage at time $t$, and similarly a penalty of $\lambda_{O}$ units if a new outage-period begins at time $t$. Delivery of a packet of quality $q$ incurs a cost of $\lambda_{q}$.

Because the probability distribution of $l(t+1)$ is completely determined by the value of $l(t)$, the action $u(t)=\left(q(t),E(t)\right)$ chosen at time $t$, and the unit step cost can be expressed solely as a function of $l(t)$, the system state can be taken to be $l(t)$. The problem is thus a Markov Decision Process (MDP) involving only a finite number of actions and states, and the cost is optimized by a stationary Markov policy~\cite{a12}. 

We have to solve,
\begin{align}\label{scp1}
\min_{\pi} \bar{C} + \lambda_{E}\bar{E}.
\end{align}
Denote by $\pi^{\star}_n(\lambda_E)$, the optimal policy which solves the single client problem. 
We also let
\begin{align}\label{scopol}
V_n(\lambda_E) = \min_{\pi_n} \left\{\bar{C} + \lambda_{E}\bar{E}\right\},
\end{align}
be the optimal cost, and $V_n(\lambda_E,\pi)$ be the cost associated with a policy $\pi$ when the power usage is priced at $\lambda_E$.  
\section{Threshold Structure of the Optimal Policy for the Single Client Problem}
We next show that the optimal policy for the single client problem has a certain ``threshold structure". Precise definition will be given shortly. Roughly this means that the transmission power should increase, and quality of transmitted packet decrease as the buffer level of video playtime decreases. Such a property is quite appealing because the decision process is simple to describe. Moreover, it has an added computational advantages since it reduces the size of the policy search space.

In the below, we omit the subscript $n$. Our approach to showing the threshold structure of the optimal policy will be to analyze the corresponding $\beta$-discounted optimization problem, and show that the solution to it has threshold structure. The result for the undiscounted problem then follows straighforward by using results in~\cite{blackwell1}.

We begin with a discussion of the $\beta\in\left(0,1\right)$ discounted infinite horizon cost problem for the single client. Let
\begin{align}\label{disc}
& V_\beta(x) = \min_{\pi}\lim_{t\to\infty}\mathbb{E}\left[ \sum_{t=0}^{\infty}\beta^t\left(C(l(t),u(t))+\lambda_E E(t) \right)\right],
\end{align} 
be the minimum $\beta$-discounted infinite horizon cost for the system starting in state $x$ at time $0$, where $x$ can assume values in the set $\{0,1,\ldots,B\}$, and where $C(l,u)$ is the one-step cost associated with chooing a scheduling action $u$ when the buffer level is $l$.

Similarly let $V^s_\beta(x)$ denote the minimum discounted cost incurred in $s$ time-slots when the starting value of the state is $x$,
\begin{align*}
V^s_\beta(x) = \min_{\pi^s}\mathbb{E}_{x}\left[ \sum_{t=0}^{s}\beta^t\left(C(l(t),u(t))+\lambda_E E(t) \right)\right],
\end{align*} 
where $\pi^s$ is a policy for the $s$ horizon $\beta$-discounted problem. The functions $ V_\beta(x),V^s_\beta(x)$ are not to be confused with their undiscounted counterparts $V_n(\lambda_E)$ that were defined in the previous section.

Th one-step Dynamic programming backward induction can be written as,
\begin{align}\label{eq:1}
&V^{s}_\beta(x)  = \min_{(q,E)}1(x=0)+ \lambda_E E\notag \\
&+P(q,E)\left[\lambda_q+\beta V^{s-1}_\beta(\mathcal{S}(x)) \right]\notag\\
&+\left(1-P(q,E)\right)\left[1(x=1)\lambda_O+\beta V^{s-1}_\beta(\mathcal{F}(x)) \right]\notag\\
& =1(x=0) + 1(x=1)\lambda_O+ \left[\beta V^{s-1}_\beta(\mathcal{F}(x)) \right]\notag\\
&+ \min_{\boldsymbol{u}}  \{\hat{C}(\boldsymbol{u})-P(\boldsymbol{u}) D^\beta_s(x)\},
\end{align}
where
\begin{align}
D^\beta_{s}(x):&=1(x=1)\lambda_O + \beta\left\{ V_\beta^{s-1}(\mathcal{F}(x))-V_\beta^{s-1}(\mathcal{S}(x))\right\},\label{df}\\
& s=1,2,\ldots,\notag
\end{align}
\begin{align}
\mbox{ while }\hat{C}(\boldsymbol{u}):=\lambda_E E + P(q,E)\lambda_q,\label{cost}
\end{align}
is defined to be the one-step cost associated with choosing the action $\boldsymbol{u}=(q,E)$, which is further composed of two terms i) cost for using power of amount $E$, and ii) the disutility associated with delivering quality $q$ video packet.

We will assume that the packet transmission success probability $P(q,E)$ is
\begin{enumerate}
\item increasing in $q$ for a fixed value of power $E$, i.e., for a fixed value of transmission power, a reduction in video quality increases the data transmission rate, or equivalently increases the chances of successful delivery of packets.
\item increasing in $E$ for a fixed value of video quality $q$.
\end{enumerate}
\begin{definition}[Threshold Policy]
We say a policy is of threshold-type if it satisfies the following for each stage $s$:
\begin{itemize}
\item Fix any $E \in \{\hat{E}_1,\hat{E}_2,\ldots,\hat{E}_n\}$. If the policy chooses the action $\left(q,E\right)$ in state $x$, then it does not choose the actions $\{\left(\hat{q},E\right): \hat{q} < q \}$ for any state $1\leq y\leq x$. Put differently, for a fixed choice of transmission power, a threshold policy does not switch to a higher video quality if the buffer level is reduced.             
\item Fix any $q \in \{Q_1,Q_2,\ldots,Q_n\}$. If the policy chooses the action $\left(q,E\right)$ in state $x$, then it does not choose the actions $\{(q,\tilde{E}): \tilde{E} < E \}$ for any state $1\leq y\leq x$. For a fixed choice of video quality, the policy does not switch to a lower transmission power if the buffer level is decreased.            
\end{itemize}
\end{definition}
The following fact follows from the definition of a threshold policy.
\begin{lemma}\label{thresh}
Let $x,y \in \{1,2,\ldots,B\}$ be two values of buffer levels such that $x>y$. Let $\pi$ be a threshold policy, and denote by $u_x,u_y$ the actions that $\pi$ chooses for the state values $x$ and $y$. Then the transmission success probabilities satisfy $P(u_x)<P(u_y)$.
\end{lemma}

In the following, $(u,\pi)$ is the policy that chooses the action $u$ in the first slot (irrespective of the system state $l$), and thereafter implements the policy $\pi$. Let $V^{s,\pi}_\beta(x)$ be the discounted cost incurred by the system starting in state $x$ and operating for $s$ time-slots under the application of policy $\pi$. We have, 
\begin{lemma}\label{l2}
Let $u_1,u_2$ be two actions satisfying $P(u_2)>P(u_1)$. Then,
\begin{align*}
& V^{s,(u_2,\pi^{\star})}_\beta(\mathcal{F}(x))-V^{s,(u_1,\pi^{\star})}_\beta(\mathcal{S}(x)) = \\
& P(u_1) \left\{\beta V^{s-1}_\beta(\mathcal{S}(\mathcal{F}(x)))-V^{s-1}_\beta(\mathcal{S}(\mathcal{S}(x)))\right\} \\
&+(1-P(u_2)) \left\{  1(\mathcal{F}(x)=1)\lambda_O + \beta V^{s-1}_\beta(\mathcal{F}(\mathcal{F}(x)))\right. \\
& \left. -V^{s-1}_\beta(\mathcal{F}(\mathcal{S}(x)))\right\}+\hat{C}(u_2)-\hat{C}(u_1)\\
&= P(u_1) \left\{ \beta V^{s-1}_\beta(\mathcal{F}(\mathcal{S}(x)))-V^{s-1}_\beta(\mathcal{S}(\mathcal{S}(x)))\right\} \\
&+(1-P(u_2)) \left\{ 1(\mathcal{F}(x)=1)\lambda_O + \beta V^{s-1}_\beta(\mathcal{F}(\mathcal{F}(x)))\right. \\
& \left. -V^{s-1}_\beta(\mathcal{S}(\mathcal{F}(x)))\right\}+\hat{C}(u_2)-\hat{C}(u_1).
\end{align*}
 \end{lemma}
 The following result is crucial to show the threshold nature of policy.
\begin{lemma}\label{l3}
The functions $D^\beta_{s}(x),x=1,2,\ldots,B-T+1$ are decreasing in $x$ for each $s=1,2,\ldots$.
\end{lemma}
\begin{proof}
Within this proof, let $\pi_s^{\star}$ be the optimal policy for the $\beta$-discounted $s$ time-slots problem, and let $(u,\pi_{s-1}^{\star})$ be the policy for $s$ time-slots which takes the action $u$ at the first time-slot, and then follows the policy $\pi_{s-1}^{\star}$. In order to prove the claim, we will use induction on $s$, the number of time-slots.

Let us assume that the statement is true for the functions $D^\beta_{z}(x)$, for all $z\leq s$. In particular the function, 
\begin{align}\label{assum1}
1(x=1)\lambda_O + \beta\left\{ V_\beta^{s-1}(\mathcal{F}(x))-V_\beta^{s-1}(\mathcal{S}(x))\right\}, 
\end{align}
is decreasing for $x\in \{1,2,\ldots,B-T+1\}$. 

First we will prove the decreasing property for $x\in \{2,3,\ldots,B-T+1\}$.
Now the assumption~\eqref{assum1} made above, and~\eqref{eq:1}, together imply that $\pi_s^{\star}$ is of threshold-type.

 Fix an $x \in \{1,2,\ldots,B-T\}$ and denote by $u_1,u_2,u_3,u_4$, the optimal actions at stage $s$ for the states $\mathcal{S}(x),\mathcal{F}(x),\mathcal{S}(x+1),\mathcal{F}(x+1)$ respectively. Note that the threshold nature of $\pi_s^{\star}$ implies that, 
\begin{align*}
&P(u_1)<P(u_2), P(u_3)<P(u_4)\mbox{ and },\\
&P(u_3)<P(u_1), P(u_4)<P(u_2).
\end{align*}
This is true because as the value of state decreases in the interval $\{1,2,\ldots,B\}$, a threshold policy switches to an action that has a higher transmission success probability. So it follows from Lemma~\ref{l2} that
\begin{align*}
&V_{\beta}^{s}(\mathcal{F}(x+1))-V_{\beta}^{s}(\mathcal{S}(x+1))\\
&\leq V_{\beta}^{s,(u_2,\pi_{s-1}^{\star})}(\mathcal{F}(x+1))-V_{\beta}^{s}(\mathcal{S}(x+1))\\
& =\hat{C}(u_2) -\hat{C}(u_3) \\
&+ P_{c}(u_3)\times \beta \left[V_\beta^{s-1}(\mathcal{F}(\mathcal{S}(x+1)))-V_\beta^{s-1}(\mathcal{S}(\mathcal{S}(x+1)))\right]\\
&+\left(1-P_{c}(u_2)\right)\times \\
& \left\{1(\mathcal{F}(x+1)=1)+ \beta V_\beta^{s-1}(\mathcal{F}(\mathcal{F}(x+1))) \right. \\
&\qquad\qquad\left. - V_\beta^{s-1}(\mathcal{S}(\mathcal{F}(x+1))) \right\}\\
&\leq \hat{C}(u_2) -\hat{C}(u_3) \\
&+ P_{c}(u_3)\times \beta \left[V_\beta^{s-1}(\mathcal{S}(\mathcal{F}(x)))-V_\beta^{s-1}(\mathcal{S}(\mathcal{S}(x)))\right]\\
&+\left(1-P_{c}(u_2)\right)\times \\
&\left[1(\mathcal{F}(x)=1)+ \beta V_\beta^{s-1}(\mathcal{F}(\mathcal{F}(x))) - V_\beta^{s-1}(\mathcal{S}(\mathcal{F}(x)))\right]\\
&\leq V_{\beta}^{s}(\mathcal{F}(x))-V_{\beta}^{s}(\mathcal{S}(x)),
\end{align*}
where the first inequality follows since a sub-optimal action in the state $\mathcal{F}(x+1)$ increases the cost-to-go for $s$ time-slots, the second inequality is a consequence of the assumption that the functions $V_\beta^{s-1}(\mathcal{F}(x))-V_\beta^{s-1}(\mathcal{S}(x))$ are decreasing in $x$, while the last inequality follows from the fact that a sub-optimal action in the state $\mathcal{S}(x)$ will increase the cost-to-go for $s$ time-slots. Thus we have proved the decreasing property of $D^\beta_{s+1}(\cdot)$ for $x\in \{2,3,\ldots,B-T+1\}$, and it remains to show that $D^\beta_{s+1}(1)>D^\beta_{s+1}(2)$.

Once again, let $u_1,u_2,u_3,u_4$ be the optimal actions at stage $s$ for the states $T,0,T+1,1$ respectively. Using the same argument as above (i.e., assuming that the actions taken in stage $s$ at states $T,T+1$ are the same, and the actions taken in the states $0,1$ are the same), it follows that
\begin{align*}
& D_{s+1}(1)-D_{s+1}(2)\geq\\
& \qquad \left(1+\lambda_O -\beta \lambda_O \right)- \left(V^s_\beta(T)-V^s_\beta(T+1)\right).
\end{align*}
However, then $V^s_\beta(T)-V^s_\beta(T+1) \leq 1+\lambda_O -\beta \lambda_O$ (for $s$ stages, apply the same actions for the system starting in state $T$, as that for a system starting in state $T+1$, and note that the two systems couple at a stage $t-1$, when the latter system hits the state $1$ at any stage $t$; the hitting stage is of course random). This gives us,
\begin{align*}
& D_{s+1}(1)-D_{s+1}(2)\geq 0,
\end{align*}
and thus we conclude that the function $D_{s+1}(x)$ is decreasing for $x\in \{1,2,\ldots,B\}$. In order to complete the proof, we notice that for $s=1$, we have,
\begin{align*}
D_1^{\beta}(x) = 1(x=1)\lambda_O,
\end{align*}
and thus the assertion of Lemma is true for $s=1$.
\end{proof}
\begin{figure}[!t]
	\centering
	\includegraphics[width=0.5\textwidth]{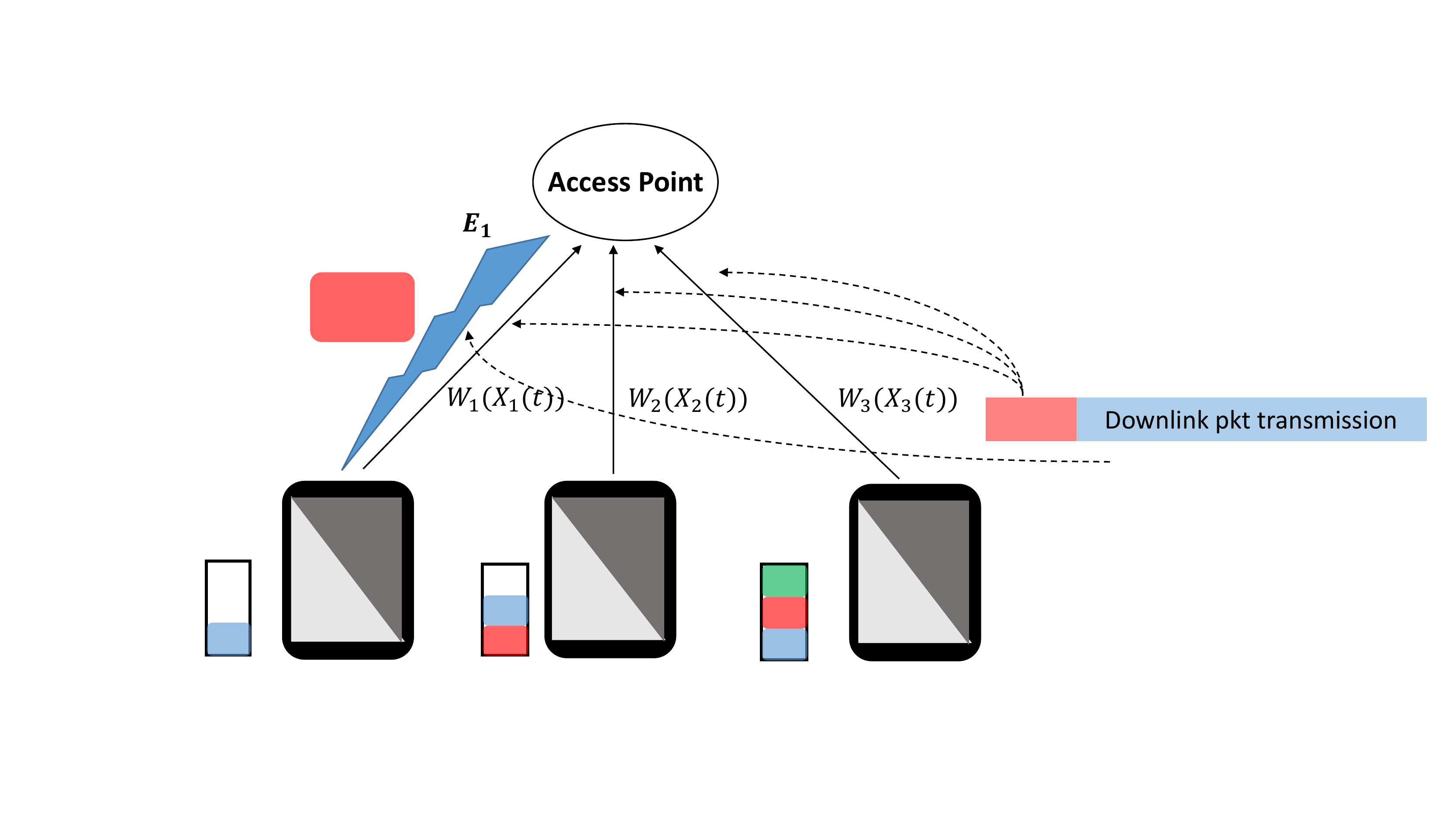}
	\caption{Time-slot is divided into two parts. The clients declare their indices $W_n(X_n(t))$ in the first part. Second part involves AP transmitting packet for the client with the largest value of the index.}
	\label{fig2}
\end{figure}
\begin{theorem}\label{t1}
For the single client scheduling problem~\eqref{scp1} there is a threshold policy that is Blackwell optimal~\cite{blackwell}, i.e., it is optimal for the discounted cost problem~\eqref{disc} for all values of discount parameter $\beta\in (\hat{\beta},1)$ for some $\hat{\beta}\in (0,1)$. Moreover this policy is also optimal for the average cost problem~\eqref{scp1}. Thus $\pi_n^{\star}(\lambda_E)$ is of threshold-type.
\end{theorem}
\begin{proof}
Fix a $q$ and let $E_i,E_j, i>j$ be two power levels. Without loss of generality, let $u_1=(q,E_i),u_2=(q,E_j)$. Clearly $\hat{C}(u_1)>\hat{C}(u_2)$~\eqref{cost}. In the Bellman equation~\eqref{eq:1}, consider the term depending on $\bs{u}$, i.e. the term $\hat{C}(u)-P(u) D^\beta_s(x)$. For $x,y\in \{1,2,\ldots,B-T+1\}$, $x>y$, we have,
\begin{align*}
& \hat{C}(u_1)-P(u_1) D^\beta_s(x) - \left(\hat{C}(u_2)-P(u_2) D^\beta_s(x)\right)\\
&-\{\hat{C}(u_1)-P(u_1) D^\beta_s(y) - \left(\hat{C}(u_2)-P(u_2) D^\beta_s(y)\right)\}\\
&=\left(P(u_1)  - P(u_2)\right) \left(D^\beta_s(y)-D^\beta_s(x)\right)\\
&\geq 0,
\end{align*}
where the last inequality follows from Lemma~\ref{l3}. Thus it follows that if action $u_1$ is preferred over action $u_2$ for any state $x$, then $u_1$ will also be preferred over action $u_2$ for any state $y<x$, $y\in \{1,2,\ldots,B-T+1\}$. Finally note that it follows from the Bellman equation~\eqref{eq:1} and~\eqref{sf}, that the optimal action for states $x>B-T+1$ is to let $E=0$ (since any packet that is received will be lost due to buffer over flow). The proof for variations in power levels is similar. Thus it follows from the definition of a threshold policy that the optimal policy is of threshold type.

Finally note that the statement regarding Blackwell optimality follows from the result in the above paragraph, and because the state-space is finite.
\end{proof}
\section{Solution of Primal MDP}
We now utilize the solution of the dual MDP and present the solution of the Primal Problem which involves minimizing the net operation cost~\eqref{costavg} subject to average power constraints~\eqref{pmdp}.

\begin{theorem}\label{t2}
Consider the Primal MDP~\eqref{pmdp} and its associated dual problem defined in~\eqref{dual}. There exists a price $\lambda^\star_E$ such that $\left(\pi^{\star}(\lambda^\star_E),\lambda^\star_E\right)$ is an optimal primal-dual pair and thus the policy $\pi^{\star}(\lambda^\star_E)$ solves the Primal MDP. 
\end{theorem}
\begin{proof}
We observe that there is a one-to-one correspondence between any stationary randomized policy, and the measure it induces on the state-action space, and thus the Primal MDP can be posed as a linear program~\cite{ergodic1,ergodic2}. Thus it follows from Slater's condition~\cite{Ber87} that for the Primal MDP, strong duality holds if there exists a policy $\pi$ that satisfies the constraints $\sum_n\bar{E}_n < \bar{E}$. However the policy which never schedules any packets incurs a net power expenditure of $0$, and thus Slater's condition is true for the Primal MDP if $\bar{E}>0$. The claim of the Theorem then follows from Lemma~\ref{l3}.
\end{proof}
\emph{$\pi^{\star}(\lambda^\star_E)$ is Decentralized} : Since the optimal decision $u_n(t)=\left(q_n(t),E_n(t)\right)$ for client $n$ at time $t$ can be obtained by solving the single client  MDP~\eqref{scp1}, the decision $u_n(t)$ can be taken by client $n$ itself. The system behaves as if there are $N$ clients operating in parallel without any interaction amongst themselves. They are coupled only through the price $\lambda_E$. This eliminates the need for a centralized controller at the AP.

We still need to address a couple of issues in order to complete the discussion on obtaining optimal policy. In order to implement the optimal policy $\pi^{\star}(\lambda^\star_E)$, the clients need to know $\lambda^\star_E$, i.e., the optimal price of power. $\lambda^\star(E)$ can be solves for if the entire system parameters are known at the AP. However, we seek a decentralized solution, in which the clients do not share their parameters with the AP. This can be attained by iterating on the price $\lambda(E)$ via the sub-gradient descent method~\cite{shor} as explained below.
\subsection{Obtaining $\lambda^{\star}_E$ iteratively in a decentralized fashion}
$\lambda^\star_E$ can be obtained as the solution of the dual problem~\eqref{dualprob}. Since the dual problem~\eqref{dualprob} involves maximizing the concave function $D(\lambda_E)$, we can use the following sub-gradient descent~\cite{shor} iterations in order to converge to the optimal value,
\begin{align}\label{pud}
\lambda^{k+1}_E = \left\{\lambda^k_E - \alpha_k \frac{\partial D}{\partial \lambda_E} \right\}^{+},~k=1,2,\ldots,
\end{align}
where the sub-gradient $\frac{\partial D}{\partial \lambda_E} $ is given by,
\begin{align}\label{partial1}
\frac{\partial D}{\partial \lambda_E} =\bar{E}- \sum_n \bar{E}_n(\pi^\star_n(\lambda^k_E)),
\end{align} 
where $\bar{E}_n(\pi^\star_n(\lambda^k_E))$ is the average power consumption of client $n$ when the AP sets the price at the value $\lambda^k_E$, and each client $n$ employs the corresponding optimal policy $\pi^\star_n(\lambda^k_E)$. The term $\sum_n \bar{E}_n(\pi^\star_n(\lambda^k_E))$ represents in a certain sense the total ``congestion" faced by the AP when it sets the price at $\lambda^k_E$. 

The iterations~\eqref{partial1} yield a decentralized solution. At each iteration $k$, AP declares the price $\lambda^k_E$. The clients then solve for their optimal policies $\pi^\star_n(\lambda^k_E)$, which also yields the average power consumptions $\bar{E}_n(\pi^\star_n(\lambda^k_E))$. They then declare these values $\bar{E}_n(\pi^\star_n(\lambda^k_E))$ to the AP. The AP then updates the price to $\lambda^{k+1}_E$ using~\eqref{pud}. The scheme is summarized in Figure~\ref{fig3}.
\begin{figure}[!t]
	\centering
	\includegraphics[width=0.5\textwidth]{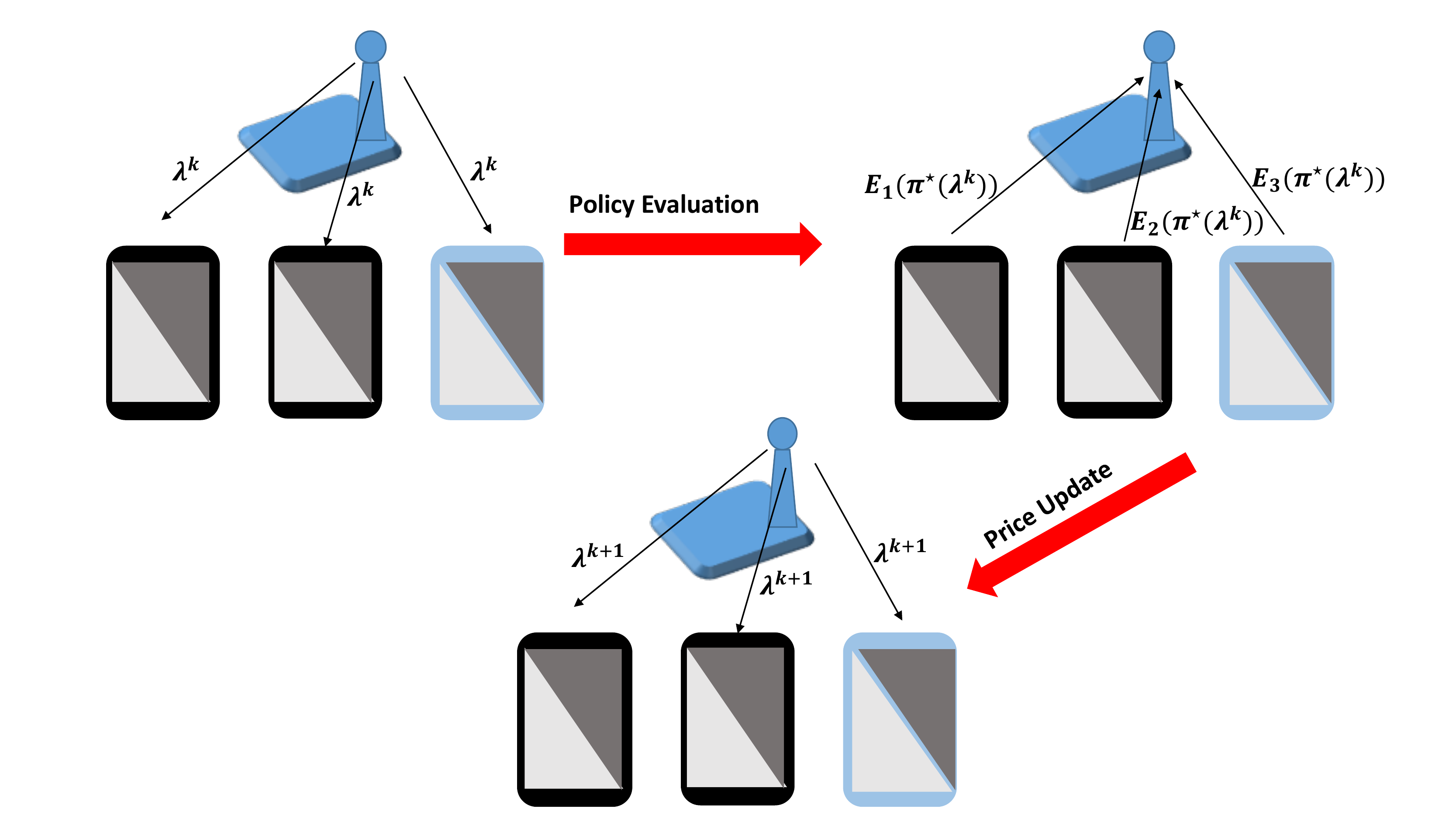}
	\caption{Decentralized iterations involving policy evaluations (at clients) followed by price updates (at AP) which converge to optimal price $\lambda^\star$ and optimal policy $\otimes \pi^\star(\lambda^\star_E)$.}
	\label{fig3}
\end{figure}
\subsection{Online Learning}
In case either the clients do not know their (possibly time varying ) system parameters such as channel reliabilities $P(q,E)$, buffer sizes etc., or they do not want to solve for the optimal policy $\pi^\star(\lambda^k_E)$, we would like to bypass the policy calculation step of iterations~\eqref{pud},~\eqref{partial1} (see Fig.~\ref{fig3}). Our problem thus lies in the realms of Reinforcement Learning~\cite{sutton}, and we can use online learning schemes such as Q-learning etc.~\cite{sutton}. These schemes adaptively ``learn" the optimal policy by simultaneoulsy ``exploring" the policy space, and ``exploiting" the past experience in the form of collected rewards (video streaming experienced so far). It involves them to keep track of the Q-factors $Q(l,u)$, which represent the ``positive effect" of choosing the action $u$ when the buffer level is $l$.  The Q factors should not be confused with the video quality. The $Q$ values are updated as,
\begin{align*}
 &Q^{t}(l(t),u(t)) (1-\beta_t) + \beta_t \left\{C(t)+\min_{u} Q^t(l(t+1),u)\right\}\\
 &\leftarrow Q^{t+1}(l(t),u(t)),
\end{align*}
where $\beta_t$ is the learning rate, and the action chosen at time $t$ is the one which greedily minimizes the cost $Q^t(l(t),u)$ of taking action $u$ when system state is $l(t)$, i.e.,
\begin{align*}
u(t) \in \arg\min_{u} Q(l(t),u). 
\end{align*}
The Q-learning iterations stated above can be combined with the ``price learning iterations"~\eqref{pud} by performing the latter on a slower time-scale, i.e., letting $\alpha_t=o(\beta_t)$. The resulting algoeithm, which simultaneouly learns the optimal price and optimal policies is guaranteed to converge to the optimal values~\cite{borkarscale,borkarbook}.
\section{Streaming under Hard Constraints}\label{sec:pp}
We now consider an important extension of the video streaming problem which involves making decisions under  ``hard constraints" imposed on the scheduling actions chosen at each time $t$ such as i) the number of orthogonal channels that can be utilized for transmission at each time $t$, and ii) the peak power that can be utilized for packet transmissions. We will discuss derive algorithms only for i), since the algorithms for ii) can be derived analogously.   

Depending upon the complexity of the underlying decision process, we classify the problem instances that involve hard constraints into two distinct categories.
\begin{enumerate}
\item  Pritoritizing Clients: In this case, if a client is chosen for scheduling, then there is only a single option for power-level and video quality. Thus, the problem of the scheduler is to choose $M$ clients in each time slot for packet transmission, or in other words, prioritize the clients for packet transmissions. Here $u_n(t)=1$ or $0$ accordingly whether client $n$ was chosen or not. 
\item Quality-Power Adaptation: The scheduler also has to make dynamic quality-power adaptations on top of prioritizing the clients for packet transmissions as in i). This is the more general set-up in which the action for client $n$ is given as $u_n(t)=(q_n(t),E_n(t)$. 
\end{enumerate}
In either of the above set-up, we will propose easily implementable and simple to compute index based policies.
\section{Prioritizing Clients}\label{sec:pc}
The scheduling has to be performed under the constraint on the number of orthogonal channels available for transmission. The set-up is thus equivalent to the restless multiarmed bandit problem (RMABP)~\cite{whittle,mahajan_bandit}. It is well known that the Whittle's Index policy~\cite{whittle,Whittle2011Book} is asymptotically optimal for the RMABP in the limit the number of client $N$ and the number of orthogonal channels $M$ are scaled to $\infty$, while keeping their ratio $M/N$ a constant~\cite{weiss}. 

We now briefly describe the notion of indexability and introduce Whittle's Index Policy. Consider the following single-client MDP parameterized by the ``transmission power price" $\lambda$,
\begin{align}\label{scpr}
\min_\pi \bar{C}_n + \lambda \bar{U}_n.
\end{align}
In this modified problem, a client is charged a price of $\lambda$ units if it utilizes one unit of power, and the system evolution of the client proceeds exactly as described in earlier sections. The transmission price roughly corresponds to the minimum price that should be charged in order that the net utilization of the bandwidth matches the available bandwidth.
Let $S_n(\lambda)$ denote the set of state values of the single client, for which the optimal action is to choose the action $U(t)=0$, i.e. not transmit. If the set $S_n(\lambda)$ is non-decreasing in the price $\lambda$, i.e., 
\begin{align*}
\lambda_1\leq \lambda_2 \implies S_n(\lambda_1) \subseteq S_n(\lambda_2),
\end{align*}
then the single client relaxed problem~\eqref{scpr} is indexable. The original problem with hard constraint~\eqref{hardmdp} is indexable if each of the $N$ individual single client problems are indexable. If the problem is indexable, then define,
\begin{align}
W_n(l) = \inf_{\lambda} \{l\in S_n(\lambda) \}.
\end{align}
At each time $t$, Whittle's Index policy activates $M$ clients having the largest indices $W_n(l(t))$. It follows from Theorem~\ref{t1}, that the scheduling problem of choosing $M$ clients for packet transmissions is indexable.
\begin{theorem}
The problem of optimizing video streaming experience under the availability of $M$ orthogonal channels is indexable.
\end{theorem}
We next provide an algorithm for computing the Whittle's Indices.

\emph{Calculating Whittle's Indices}
If the channel reliabilities $p_n$ are assumed to be known to the clients, then each client $n$ can solve the following system of linear equations for $A,V(x),\lambda$,
\begin{align*}
A + V(x) &= C(x,0) + V(\mathcal{F}(x)),~ x>k\\
A + V(x) &= C(x,1) + \lambda + p\mathcal{S}(x)+(1-p)V(\mathcal{F}(x)),~ x<k\\
V(0)&=0,\\
A+ V(k) &= C(k,0) +  \mathcal{F}(k), \\
A+V(k) &= C(k,1) + \lambda + pV(\mathcal{S}(k))+(1-p)V(\mathcal{F}(k)),\\
\end{align*}
where $A$ is the average cost, $V(x)$ is the relative cost for system starting in state value $x$, while the last two equations state that the active and passive actions are both optimal when the price is set to $\lambda$.
\section{Quality-Power Adaptation }\label{sec:qpa}
The extension of the MABP setting to the case when more than two actions are available, is called bandit superproceses~\cite{Whittle2011Book,mahajan_bandit}. For superprocesses, the notion of Whittle's index policy does not exist~\cite{mahajan_bandit}. Hence we will define appropriate notion of indices, and provide the corresponding index policies. We will assume that the value function $V(\cdot)$ corresponding to the scheduling problem is separable, i.e., there exist functions $g_n(\cdot), n=1,2,\ldots,N$ such that,
\begin{align*}
V(l_1,l_2,\ldots,l_N) = \sum_{n}g_n(l_n), 
\end{align*}
for all possible values of the state vector $(l_1,l_2,\ldots,l_N)$. Under the separable value function assumption, the optimal policy chooses $M$ clients in the decreasing order of the indices, 
\begin{align*}
\max_{u}C_n(l_n,u) + P_{n}(u)V_n(\mathcal{S}(l_n)) + (1-P_n(u))V_n(\mathcal{F}(l_n)),
\end{align*}
providing us a convenient index based policy. 

Since we would like to use online data from system operation in order to compute the above indices, we propose a Q-learning type algorithm. 
\begin{algorithm}
\caption{Q-learning based Index Policy}
At each time $t$, for each client $n$ maintain Q-values $Q^t_n(l,u)$.
\begin{enumerate}
\item \emph{Scheduling} : Implement the action $(n,u)$ with a probability \begin{align*}
\frac{\exp(\tau Q^t_n(l,u))}{\sum_{m,l,u}\exp(\tau Q^t_n(l,u))}.\end{align*}
\item \emph{Q-Update}: Let client $n$ be served in time $t$. Then the update occurs as,
\begin{align*}
 Q^{t+1}_{n}(l_{n}(t),u_{n}(t)) &= Q^{t}_{n}(l_{n}(t),u_{n}(t)) \left(1-\alpha_t\right) \\
&+ \alpha_t \left\{ C_{n}(t) + \max_{u} Q^{t}_{n}(l_{n}(t+1),u)\right\}
\end{align*}
\end{enumerate}
\end{algorithm}
\section{Extensions}\label{tqv}
Several useful extensions can be considered. We briefly describe the set-up and mention the approach to designing algorithm. 

\emph{Temporal Variations} in video quality are also an important factor that affect the user engagement~\cite{yim}. Infact~\cite{yim} shows that it might play an even more role than the video quality. In order to optimize over temporal variations, we simply need to augment client $n$'s state $l_n(t)$ by including $q_{n}^{-}(t)$, the video quality of the latest packet delivered to client $n$ uptil time $t$. Index policies can then be derived using the augmented state variable.

\emph{Multi-Hop Networks }: Video packets have to traverse a path comprising of multiple nodes before they reach the end user. The network is described by graph $G=(V,E)$, where $V$ is the set of nodes, and $E$ is the set of directed links that can be used for packet transmissions. The link capacities $C_\ell$ are stochastic processes. Each node $i\in V$ maintains a packet buffer for each client whose packets are routed through node $i$. The state of client $n$, i.e., $l_n(t)$ is now described by the $E$ dimensional vector $\{l_{n,i}(t)\}_{i\in V}$, where $l_{n,i}(t)$ is the buffer occupancy of client $n$ at node $i$ at time $t$. In case a client does not use a node $i$ for routing packets, one can simply set the corresponding buffer to be zero. 

Within this set-up, two important class of policies can be considered. 
\begin{enumerate}
\item Client-Level Policies: Network operator charges price at the rate $\lambda_\ell$ for using bandwidth at link $\ell$. The vector $\lambda\geq0$ comprising the prices $\lambda_\ell$ are known to the clients. Client $n$ then employs the policy $\pi^\star_n(\lambda)$ that optimizes its individual operating cost $\bar{C}_n+ \sum_{\ell}\lambda_\ell \bar{E}_{n,\ell}$, where $\bar{E}_{n,\ell}$ represents the average bandwidth utilized by it on link $\ell$. Note that $\pi_n$ maps the vector $l_n(t)$ the state of flow $f$, so that the policy needs to know the network wide state of client $n$. An index policy similar to that provided in Sections~\ref{sec:pc} and~\ref{sec:qpa}.

\item Node-Level Policies: Client-Level Policies require obtaining the ``global" state of client $n$, which might be too much of a communication overhead. We can restrict ourselves to the class of decentralized policies, so that the scheduling decisions at a node $i$ depend only on $l_{n,i}(t)$, i.e., the $i$-th component of client state. The policies can be modified in a straightforward fashion to yield the optimal policy within this class.
\end{enumerate}
We also mention that various function approximation techniques can be utilized~\cite{konda1999actor,bertsekas1995neuro} in order to combat the resulting state space explosion.

We would also like to mention that we can consider the extension where the clients join and leave the system dynamically.
\section{Fading Channels}\label{fading}
The results in the previous sections can be extended in a straight forward manner to the case of fading channels. Let the channel conditions for client $n$ be described by a Markov process evolving on finitely many states $\{1,2,\ldots,C_n\}$ having a transition matrix $\Pi_n$. The state of client $n$ is described by the vector $\boldsymbol{x}_n(t):=\left(l_n(t),c_n(t)\right)$, where $l_n(t)$ is the play-time duration of the packets present in the buffer at time $t$, and $c_n(t)$ is the channel condition at time $t$. 
If the client $n$ is scheduled a packet transmission of quality $q$ at an power $E$ at time $t$, then the system state at time $t+1$ is $\left(\mathcal{S}(l(t)),\tilde{c}\right)$ with a probability $P_{n,c_n(t)}(q,E)\Pi(c_n(t),\tilde{c})$, while it is $\left(\mathcal{F}(l(t)),\tilde{c}\right)$ with a probability $P_{n,c_n(t)}(q,E) \Pi(c_n(t),\tilde{c})$.

However now the cost associated to an action $\bs{u}$ also depends on the channel condition, i.e.,
\begin{align}\label{cost1}
C_c(\boldsymbol{u}):=\lambda_E E + P_{c}(l,E)\lambda_q ,
\end{align}
and a threshold policy will have a threshold structure for each value of channel condition (as defined in Section~\ref{scp}). 

\bibliographystyle{IEEEtran}
\bibliography{../../GTS/combinedbib}
\end{document}